\documentclass[a4paper]{article}
\usepackage[all]{xy}\usepackage[latin1]{inputenc}        
\usepackage[dvips]{graphics,graphicx}
\usepackage{amsfonts,amssymb,amsmath,color,mathrsfs, amstext}
\usepackage{amsbsy, amsopn, amscd, amsxtra, amsthm,authblk}
\usepackage{enumerate,algorithmicx,algorithm}
\usepackage{algpseudocode}
\usepackage{upref}
\usepackage{geometry}
\usepackage[displaymath]{lineno}
\usepackage{float}
\usepackage{bm}
\usepackage{caption}
\usepackage{yhmath}
\usepackage{booktabs}
\usepackage{subcaption}
\usepackage{multirow}
\usepackage{makecell}
\usepackage[colorlinks,
            linkcolor=red,
            anchorcolor=red,
            citecolor=red
            ]{hyperref}

\numberwithin{equation}{section}

\newcommand{\bbR}{\mathbb{R}}

\newtheorem{theorem}{Theorem}[section]
\newtheorem{lemma}{Lemma}[section]

\newtheorem{remark}{Remark}[section]

\numberwithin{equation}{section}

\begin{document}

\title{A random-batch Monte Carlo method for many-body systems with singular kernels}

\author[1,2]{Lei Li\thanks{leili2010@sjtu.edu.cn}}
\author[1,2]{Zhenli Xu\thanks{xuzl@sjtu.edu.cn}}
\author[1]{Yue Zhao\thanks{sjtu-15-zy@sjtu.edu.cn}}
\affil[1]{School of Mathematical Sciences, Shanghai Jiao Tong University, Shanghai, 200240, P. R. China}
\affil[2]{Institute of Natural Sciences and MOE-LSC, Shanghai Jiao Tong University, Shanghai, 200240, P. R. China}

\date{}
\maketitle

\begin{abstract}
We propose a fast potential splitting Markov Chain Monte Carlo method which costs $O(1)$ time each step
for sampling from equilibrium distributions (Gibbs measures) corresponding to particle systems with singular interacting kernels.
We decompose the interacting potential into two parts, one is of long range but is smooth, and the other one is of short range
but may be singular. To displace a particle, we first evolve a selected particle using the stochastic differential equation (SDE)
under the smooth part with the idea of random batches, as commonly  used in stochastic gradient Langevin dynamics.
Then, we use the short range part to do a Metropolis rejection. Different from the classical Langevin dynamics, we only run the SDE dynamics with
random batch for a short duration of time so that the cost in the first step is $O(p)$, where $p$ is the batch size.
The cost of the rejection step is $O(1)$ since the interaction used is of short range. We justify the proposed random-batch Monte Carlo method, which
combines the random batch and splitting strategies, both in theory and with numerical experiments. While giving comparable results
for typical examples of the Dyson Brownian motion and Lennard-Jones fluids, our method can save more time when compared to the classical Metropolis-Hastings
algorithm.


\end{abstract}

\section{Introduction}\label{sec:intro}

The statistics of many-body systems such as pressure, energy and entropy at equilibrium can describe the physical properties
of the systems \cite{lifshitz2013statistical}. The values of these quantities correspond to expectations of some functions under
the equilibrium, which is often approximated by using Monte Carlo sampling methods \cite{allen1987,frenkel2001understanding}.
The Markov Chain Monte Carlo (MCMC) methods \cite{gilks1995markov,gamerman2006markov} are among the most popular Monte Carlo methods.
By constructing Markov chains that have the desired distributions to be the invariant measures, one can obtain samples from
the desired distributions by recording the states of the Markov chains.
A typical MCMC algorithm is the Metropolis-Hastings algorithm \cite{metropolis1953,hastings1970monte}.

In this work, we revisit this classical problem, namely sampling from equilibrium distributions for many-body interacting
particle systems. Suppose that we want to sample from the $N$-particle Gibbs distribution (or generally the Boltzmann distribution) \cite{lifshitz2013statistical}
\begin{gather}\label{eq:Gibbs}
\pi(\bm{X})\propto \exp\left[-\beta H(\bm{X})\right],
\end{gather}
with $\bm{X}=(x_1,\ldots, x_N)\in \bbR^{Nd}$ ($x_i\in \bbR^d$,  and $d\ge 1, d\in \mathbb{N}$), $\beta$ being a positive constant, the $N$-body energy
\begin{gather}
H(\bm{X}):=\sum_{i=1}^N w_i V(x_i)+\sum_{i,j: i<j} w_i w_j U(x_i-x_j),
\end{gather}
and $V$ being the external potential which is assumed to be smooth. Without loss of generality, the kernel for interaction potential $U(\cdot): \bbR^d\to \bbR$ is assumed to be of the symmetric form
\begin{gather}
U(x)=u(|x|),
\end{gather}
and is possibly singular and of long range. The weight $w_i$ is often the mass or charge of the particle measuring how important
the $i$th particle is. In this work, we will only consider
\begin{gather}
w_i\equiv w,
\end{gather}
and the generalization of the proposed algorithm to nonuniform weights is straightforward.

As mentioned above, the classical algorithm to sample from the Gibbs distribution is MCMC algorithms including the Metropolis algorithm.
Later on, there are some variants like the multiple time-step MC algorithm (MTS-MC) \cite{hetenyi2002multiple} using the splitting idea. Often, these algorithms update a randomly chosen particle per Markov jump in the state space, as updating one particle per iteration can be more efficient (see Ref. \cite{frenkel2001understanding} for more explanation).
A rejection step is performed following some random movement. The acceptance probability is computed using the difference between the energies of two configurations according to the detailed balance condition. Hence, computing the acceptance probability often takes $O(N)$ time after moving a randomly chosen particle. Due to this reason, throughout this paper, we will use ``an iteration'' or ``a move'' to mean such a Markov jump together with the rejection step.

Another class of sampling methods is the Langevin sampling \cite{besag1994comments,roberts1996,dalalyan2017,raginsky2017non}. In fact,
if one moves $X^i$ according to the following stochastic differential equation (SDE) system \cite{oksendal03}
\begin{gather}\label{eq:langevinsampling}
dX^i= -\gamma\left[ w \nabla V(X^i)+w^2\sum_{j: j\neq i} \nabla U(X^i-x_j)\right] \,dt+\sqrt{\frac{2\gamma}{\beta}}\, dW^i,
\end{gather}
with $X^i(0)$ being arbitrary and $\{W^i\}$ being independent standard $d$-dimensional Wiener processes (Brownian motions) for $i=1,\cdots, N$, 
then the invariant measure is the Gibbs distribution \eqref{eq:Gibbs}. For the Langevin dynamics, one often
desires the forcing term to be of order $1$, so one may choose
\begin{gather}\label{eq:gammaexpress}
\gamma=\frac{1}{w^2(N-1)}.
\end{gather}
For the equations with a big summation on the right hand side, one may apply the mini-batch ideas \cite{robbins1951stochastic,bottou1998online,bubeck2015convex,jin2018random}  to reduce the cost.
The random mini-batch idea was famous by its application in the stochastic gradient descent (SGD) \cite{robbins1951stochastic,bottou1998online}.
Recently, it has been applied to interacting particle systems by Jin {\it et al.} yielding a fast algorithm called the random batch method (RBM) \cite{jin2018random}, which is asymptotic-preserving in the mean field limit regime. This algorithm not only computes the equilibrium distribution correctly, but also can capture the dynamics of the interacting particle systems. Building the mini-batch ideas into the Langevin
sampling leads to the stochastic gradient Langevin dynamics (SGLD) \cite{welling2011bayesian,ma2015complete},
which is a fast approximate Monte Carlo method, and has been widely used for Bayesian inference problems \cite{Martin:SISC:12}.
Recently, some particle based variational inference sampling methods have been proposed (see \cite{liu2016stein,rezende2015variational,dai2016provable}).
These methods update particles by solving optimization problems, and each iteration is expected to make progress.
The Stein variational gradient descent \cite{liu2016stein} is a typical example. Applying the random batch idea to it can yield an efficient sampling method \cite{li2019stochastic}.

In this paper,  we combine the splitting idea (as in MTS-MC) and the random mini-batches idea (as in SGLD algorithm) to get a fast MCMC sampling algorithm. We decompose the interacting potential into two parts, one is of long range but is smooth, and the other one is of  short range but may be singular. As commented in Ref. \cite[Chap. 3]{frenkel2001understanding} and mentioned above, we choose to update only one particle per move. To move a selected particle, we first evolve it using the SDE under the smooth part with the idea of random batches, as commonly  used in SGLD. Then, we use the short range part to do a Metropolis rejection. Different from SGLD, we only run the SDE dynamics with random batch for a short duration of time so that the cost in the first step is $O(p)$, where $p$ is the batch size. The cost of the rejection step is $O(1)$ since the interaction used is of short range.
The process of our method seems reminiscent of the Metropolis-adjusted Langevin algorithm (MALA) \cite{besag1994comments,roberts1996}, where a discretized SDE was run first and then a Metropolis rejection was performed. However, we point out that the Metropolis rejection in MALA is only to correct the small errors introduced by the discretization. In our approach, we have the splitting explicitly and the rejection step aims to avoid unphysical configurations where one particle is moved to the singular core of the nearby particle, which greatly improves the sampling efficiency.

Let us make a short discussion of the regimes to consider. In most cases, we focus on the molecular regime so that $w=1$ and $\beta=O(1)$, and hence the number $\sqrt{2\gamma/\beta}$ by \eqref{eq:gammaexpress} is often small.
 In the mean field regime \cite{stanley1971, georges1996, lasry2007},
one may have $w\sim N^{-1}$ (see the discussion below and the example in section \ref{subsec:db}). Each particle satisfies the following equation,
\begin{gather}
dX^i= -\left[ \nabla V(X^i)+\frac{1}{N}\sum_{j: j\neq i} \nabla U(X^i-x_j)\right] \,dt+\sqrt{\frac{2N}{\beta}}\, dW^i.
\end{gather}
As is well known \cite{jabin2017,li2019mean}, if $\beta\sim N$ and $\sqrt{2\gamma/\beta}=O(1)$, the stationary
{\it single particle} distribution can be approximated by the solution to the nonlinear Fokker-Planck equation
\begin{gather}
\nabla\cdot(\nabla V \rho)+\nabla\cdot(\rho \nabla U*\rho)+\Delta\rho=0.
\end{gather}

The rest of the paper is organized as follows. In section \ref{sec:split}, we propose the splitting strategy and give its justifications. Based on this splitting, a move of the Markov chain consists of two parts, one is the evolution of an SDE under the smooth part of the potential while the other one is the rejection under the short range part of the potential. We verify that the Gibbs distribution is the invariant measure of the constructed Markov chain. In section \ref{sec:sderandom}, we apply the idea of random mini-batches to the SDE step and obtain our sampling algorithm. Section \ref{sec:numerical} is devoted to the numerical justifications, where we compare our method with Metropolis-Hastings algorithm on the Dyson Brownian motion example and Lennard-Jones fluids. Conclusions are made in Section \ref{sec:conc}.

\section{Splitting Monte Carlo method}\label{sec:split}

In practice, the interaction kernel between a pair of particles is often singular
and possibly with long range. The typical example is the potential kernel of the following form
\begin{gather}
U(x)=\frac{1}{|x|^{\alpha}},
\end{gather}
where $\alpha>0$. See Example 1 in Section \ref{sec:numerical} for an example. Other examples include the famous Coulomb potential and the London potential \cite{Israelachvili::2010}.
If we do sampling directly using the Metropolis-Hastings algorithm \cite{hastings1970monte} (see also Algorithm \ref{alg:mh} in Section \ref{sec:numerical}),
to get a proposed move, one must consider the interaction with all other $N-1$ particles, which makes the computational cost of the algorithm $O(N)$ per move.
Moreover, due to the singular interaction, the acceptance rate can be lower.

Inspired by the MTS-MC Method \cite{hetenyi2002multiple}, we consider a splitting strategy for the binary potential $U$ such that the single-variable function $u(r)$ becomes the summation of $u_1(r)$ and $u_2(r).$ Let $U_i(x):=u_i(|x|)$ and correspondingly we have,
\begin{gather}
U(x)=U_1(x)+U_2(x),
\end{gather}
where we suppose that $U_1$ has long range but is smooth and bounded, and $U_2$ is singular and compactly supported.
The splitting is convenient because one can apply the Langevin dynamics \cite{roberts1996,dalalyan2017,raginsky2017non} for the $U_1$ part. The singular $U_2$ part can then be used for the rejection.

There are many ways to split the potential. In the adiabatic nuclear and electronic sampling Monte Carlo (ANES--MC) method \cite{martin1998novel}, the splitting is to decompose $U$ into ``fast'' and ``slow'' dynamics. In the MTS-MC method, the potential is decomposed into short-range and long-range parts, the long range part is evaluated less frequently than the short range part to improve the efficiency. In our method, we also split the potential into short-range part and long-range part, but our motivation is different from MTS-MC. In fact, we aim to build the random mini-batch idea in SGD and SGLD into the algorithm to reduce the cost, but the mini-batch method works well for smooth potentials. Naturally, we do the splitting and ask the long range part $U_1$ to be smooth and use the singular part to do rejection. In Section \ref{subsec:ljnumerics}, we will see that for Lennard-Jones potential, our method will have comparable CPU speedup with that of MTS-MC in \cite{hetenyi2002multiple} when $N=500$. Since our algorithm is $O(1)$ per iteration, we expect that the cost
to run a cycle of samples will be linear in $N$ as $N$ grows.  In our current implementation, the splitting is kind of far from the best. How to optimize the splitting to gain better spectral gap of the semigroup for the Markov chain is an interesting question left for future.

Below, we aim to justify this splitting by showing that the desired Gibbs distribution is an invariant measure of the Markov chain.

 Suppose there are $N$ particles located at $x_j$ for $j=1,\cdots, N$.  Let us consider the following method for a Markovian jump.

{\it Step 1 ---} Randomly choose a particle $i$.

{\it Step 2 ---}
 Move the particle using $U_1$ with overdamped Langevin equation:
 \begin{gather}\label{eq:sde1}
 \begin{split}
 & dX^i= -\left(\frac{\nabla V(X^i)}{w(N-1)}+\frac{1}{N-1}\sum_{j: j\neq i} \nabla U_1(X^i-x_j)\right) \,dt+\sqrt{\frac{2}{(N-1)w^2 \beta}}\, dW^i,\\
 & X^i(0)=x_i,
 \end{split}
 \end{gather}
 where $x_j$'s are fixed.
  Evolve this SDE with some time $t>0$ and obtain $X^i(t)\rightarrow x_i^*$ as a candidate position of particle $i$ for the new sample.

{\it Step 3 ---}
Use $U_2$ to do the Metropolis rejection. Define
 \begin{gather}\label{eq:acceptance}
 acc(x_i, x_i^*)=\min\left\{1, \exp\Big[-\beta \sum_{j: j\neq i}w^2(U_2(x_i^*-x_j)-U_2(x_i-x_j))\Big]  \right\}.
 \end{gather}
 With probability $acc(x_i, x_i^*)$, we accept $x_i^*$ and set
 \begin{gather}
 x_i \leftarrow x_i^*.
 \end{gather}
Otherwise, $x_i$ is unchanged. We then obtain a new sample $\{x_1,\cdots,x_N\}$ in the Markov chain.

In comparison with the multiple time-step splitting algorithm, we do not have rejection in {\it Step 2} while {\it Step 3} can be done
in $O(1)$ time since $U_2$ is local. The correctness of this splitting is guaranteed by Theorem \ref{thm:invariant}.
Before stating this theorem, we have the following lemma.

\begin{lemma}\label{lmm:detailedbalanceSDE}
The SDE \eqref{eq:sde1} has invariant measure
\begin{gather}
\pi_1(z)\propto \exp\left[-\beta\Big(w V(z)+w^2 \sum_{j:j\neq i}U_1(z-x_j)\Big) \right].
\end{gather}
Moreover, if $p(z, t; x_i)$ denotes the transition density from $x_i$, then
the detailed balance condition holds
\begin{gather}\label{eq:detailbalance}
\pi_1(x_i)p(x_i^*, t; x_i)=\pi_1(x_i^*) p(x_i, t; x_i^*).
\end{gather}
\end{lemma}
The claims here are standard in the SDE theory and we attach a proof in Appendix \ref{app:proof1} for the convenience of the readers.

Now we have the following theorem on the detailed balance of the algorithm composed of {\it Steps 1 -- 3},
namely, the Langevin dynamics step
using $U_1$ and the Metropolis rejection step using $U_2$ for a randomly chosen particle.
\begin{theorem}\label{thm:invariant}
Let $\bm{X}=(x_1,\cdots, x_N)$ be a point in the state space of the $N$ particles.
Let $T(\bm{X}, \bm{X}^*)$ be the transition probability density for {\it Steps 1--3}.
The splitting Monte Carlo satisfies the following detailed balance:
\begin{gather}\label{eq:fulldetailbalance}
\pi(\bm{X})T(\bm{X}, \bm{X}^*)=\pi(\bm{X}^*)T(\bm{X}^*, \bm{X}),
\end{gather}
where $\pi(\bm{X})$ is given by \eqref{eq:Gibbs}.
\end{theorem}

\begin{proof}

Note that the method consists of {\it Steps 1--3}, and only one particle is updated. Hence, $T(\bm{X}, \bm{X}^*)$ is nonzero only if $\bm{X}$ and $\bm{X}^*$ are different by one component at most. Hence, we only have to check two states where the $i$th component is different.

Without loss of generality, we assume $\bm{X}^*=(x_1, \ldots, x_i^*, \ldots, x_N)$ and $x_i\neq x_i^*$. Clearly,
\begin{gather}
T(\bm{X}, \bm{X}^*)=\frac{1}{N}p(x_i^*, t; x_i)acc(x_i, x_i^*),
\end{gather}
where $1/N$ is the probability for choosing the random particle.

Note that
\begin{gather}
\begin{split}
\pi(\bm{X})&=\frac{1}{Z_N}\exp\left[-\beta \Big(\sum_i wV(x_i)+\sum_{j,k: j<k}w^2 U(x_j, x_k)\Big)\right] \\
&=\frac{\pi_1(x_i)\pi_2(x_i)}{Z_3}\exp\left[-\beta \sum_{\ell: \ell\neq i} w V(x_\ell)\right]\exp\left[-\beta \sum_{j<k: j,k\neq i}w^2U(x_j-x_k)\right],
\end{split}
\end{gather}
where
\begin{gather}
\pi_2(z)\propto \exp\left[-\beta\sum_{j: j\neq i}w^2 U_2(z-x_j) \right],
\end{gather}
and $Z_N$ and $Z_3$ are some normalizing constants.
Consequently, the detailed balance condition \eqref{eq:fulldetailbalance} is reduced to
\begin{gather}\label{eq:fulldetailbalance2}
\pi_1(x_i)\pi_2(x_i)p(x_i^*, t; x_i) acc(x_i, x_i^*)=\pi_1(x_i^*)\pi_2(x_i^*) p(x_i, t; x_i^*) acc(x_i^*, x_i).
\end{gather}

Using the detailed balance \eqref{eq:detailbalance} for the SDE, Eq. \eqref{eq:fulldetailbalance2} is then reduced to
\begin{gather}
\pi_2(x_i)acc(x_i, x_i^*)=\pi_2(x_i^*) acc(x_i^*, x_i),
\end{gather}
which clearly holds by the formula of the Metropolis rejection probability \eqref{eq:acceptance}.
\end{proof}

\begin{remark}
The detailed balance \eqref{eq:detailbalance} allows us to solve the SDE on a fixed time interval, far before reaching the equilibrium. If the SDE $dX=b(X)\,dt+\sigma(X) dW$ does not have detailed balance (or $-b\pi
+\frac{1}{2}\nabla\cdot(\sigma\sigma^T \pi)\neq 0$), then {\it Step 2} should be run for long time so that $p(x, t; x_i)\approx \pi_1(x)$.  This is a problem since the Euler-Maruyama scheme can be guaranteed to converge on finite time interval, but it may not converge for infinite time interval.
\end{remark}

 \section{Random batch sampling algorithm}\label{sec:sderandom}

 Note that {\it Step 3} can be done in $O(1)$ operations using some standard data structures such as the cell list \cite[Appendix F]{frenkel2001understanding} so we want to reduce the time cost in {\it Step 2} to $O(1)$ as well.
The idea is to use the mini-batch approach, similar to the idea in SGLD \cite{welling2011bayesian,ma2015complete}.

 \subsection{Discretization of the SDE with random batch}\label{sec:sderandombatch}

We now focus on {\it Step 2}, and discretize the SDE with the Euler-Maruyama scheme \cite{kloeden2013numerical,milstein2013stochastic}.
The interaction force is approximated by that within the mini batch. This then gives the following discrete approximation.

Pick a batch size $p>1.$ Let $X^i\leftarrow x_i$.
Run the following approximation for $m$ steps ($k=0,\ldots, m-1$), where $\bm{\xi}_k$ is a random set with size $p-1$ at step $k$, which is a subset of $\{1, 2, \cdots, N\}\setminus \{i\}$:
 \begin{gather}\label{eq:randombatchEuler}
 X^i\leftarrow X^i-\tau_k \left[\frac{\nabla V(X^i)}{w(N-1)}+\frac{1}{p-1}\sum_{j\in \bm{\xi}_k} \nabla U_1(X^i-x_j)\right]+\sqrt{\frac{2\tau_k}{(N-1)w^2 \beta}} z_k,
 \end{gather}
 where  $\tau_k>0$ is a step size and $z_k\sim \mathcal{N}(0, I_d)$ with $d$ being the dimension.
 Here $\mathcal{N}(\mu, \Sigma)$ represents the multivariate normal distribution with mean $\mu$ and covariant matrix $\Sigma$.

In practice, the number $m$ is independent of $N$ and a small constant in order to reduce the algorithm complexity.
The scheme \eqref{eq:randombatchEuler} can be further written as, for $k=1,\cdots, m,$
\begin{gather}
X^i \leftarrow X^i-\tau_k\left[\frac{\nabla V(X^i)}{w(N-1)}+\frac{1}{N-1} \sum_{j: j\neq i} \nabla U_1(X^i-x_j)\right]+\sqrt{\frac{2\tau_k}{(N-1)w^2 \beta}} z_k+\tau_k\zeta_k^i,
\end{gather}
where
\begin{gather}
\zeta_k^i:=-\frac{1}{p-1}\sum_{j\in \bm{\xi}_k} \nabla U_1(X^i-x_j)+\frac{1}{N-1} \sum_{j: j\neq i} \nabla U_1(X^i-x_j).
\end{gather}

The following lemma is straightforward and we omit its proof (see Refs. \cite{hu2019diffusion} and \cite{jin2018random} for similar results).
\begin{lemma}\label{lmm:randomforce}
Random variable
$\zeta_k^i$ has zero mean, and the covariance matrix is given by
\begin{gather}
\mathrm{cov}(\zeta_k^i)=\left(\frac{1}{p-1}-\frac{1}{N-1}\right)\Lambda(x_1,\cdots, x_{i-1}, X^i, x_{i+1},\cdots, x_N),
\end{gather}
where $\Lambda$ is defined by
\[
\Lambda(x_1,\cdots, x_N) :=\frac{1}{N-2}\sum_{j: j\neq i}[\nabla U_1(x_i-x_j)-G(x_i)]\otimes [\nabla U_1(x_i-x_j)-G(x_i)]
\]
with
\[
G(x_i)=\frac{1}{N-1}\sum_{j: j\neq i}\nabla U_1(x_i-x_j).
\]
\end{lemma}

The variance of the term $\tau_k \zeta_k^i$ is thus of order $\tau_k^2$, which is small. Hence,
this gives an effective approximation for the SDE.
In our algorithm, the SGLD is performed on finite interval $[0, T]$ with $T=O(1)$ or $T\ll 1$.
 We now estimate the error in the transition probability introduced by discretizing the SDE and applying the mini-batch idea.

For convenience, we only consider constant step size in this section
$ \tau_k\equiv \tau>0.$
To justify the mini-batch discretization proposed above, we will show that the transition probabilities are close in the Wasserstein distance. We recall that the Wasserstein-$2$ distance \cite{santambrogio2015} is given by
\begin{gather}\label{eq:W2}
W_2(\mu, \nu)=\left(\inf_{\gamma \in \Pi(\mu,\nu)}\int_{\mathbb{R}^d\times\mathbb{R}^d}|x-y|^2 d\gamma\right)^{1/2},
\end{gather}
where $\Pi(\mu,\nu)$ means all the joint distributions whose marginal distributions are $\mu$ and $\nu$ respectively.

We in fact have the following Theorem \ref{thm:errestimates} for the error estimate.
In practice, we take $T=m\tau$ with $m=O(1)$ so the error is in fact of order $\tau$.
\begin{theorem}\label{thm:errestimates} Let $\lambda=1/w(N-1)$ and $D=1/\beta w^2(N-1).$
Suppose $\lambda\lesssim O(1)$. Suppose that $V$ and $U_1$ are smooth, and the derivatives up to third order are bounded. Let $Q(x, \cdot)$ be the transition probability of the SDE \eqref{eq:sde1} and $\widetilde{Q}(x, \cdot)$ be the transition probability computed by the discretization with random batch \eqref{eq:randombatchEuler}. Then,
\begin{gather}
\sup_{x}W_2(Q(x, \cdot), \widetilde{Q}(x, \cdot))\le C\sqrt{(e^{CT}-1)\left(\frac{1}{p-1}\|\nabla U_1\|_{\infty}
 \tau+(1+D^2)\tau^2\right)}.
\end{gather}
\end{theorem}

\begin{proof}
For the convenience, in this proof, we will denote the chosen $X^i$
 by $Y$ and $f_j(y):=U_1(y-x_j)$. Without loss of generality,
 we assume that $i=N$ and $j\in \{1,\ldots, N-1\}$.  Denote the $L^2(\mathbb{P})$ norm
 $\|\cdot \|:=\sqrt{\mathbb{E}|\cdot|^2}.$
To simplify the notation, we also define
 \begin{gather}
 g_1(y):=-\lambda \nabla V(y),~~\hbox{and}~~
 g_2(y):=-\frac{1}{N-1}\sum_{j=1}^{N-1}\nabla f_j(y).
 \end{gather}
The desired SDE is then given by
 \begin{gather}\label{eq:auxsde1}
 dY=g_1(Y)\,dt+g_2(Y)\,dt+\sqrt{2D}\,dW,~~t\in [0, T].
 \end{gather}

Define $ t_k:= k\tau.$
The numerical algorithm is given by
 \[
 Y^{k+1}=Y^k+g_1(Y^k)\tau-\frac{1}{p-1}\sum_{j\in \bm{\xi}_k}\nabla f_k(Y^k)\tau
 +\sqrt{2D}\Delta W_k,~~k=0,\cdots, m-1.
 \]
 To compare their distributions, we apply the standard coupling technique, namely,
 \begin{gather}
 \Delta W_k=W(t_{k+1})-W(t_k),
 \end{gather}
 where $W(\cdot)$ is the same Brownian motion used in \eqref{eq:auxsde1}.

To estimate the error of the numerical solution in comparison to the solution of the SDE, we consider the following equation for $t\in [t_{k}, t_{k+1}]$,
 \begin{gather}\label{eq:auxsde2}
 dY_{\tau}(t)=g_1(Y_{\tau}(t_k))\,dt-\frac{1}{p-1}\sum_{j\in \bm{\xi}_k}\nabla f_k(Y_{\tau}(t_k))\,dt
 +\sqrt{2D}\,dW.
 \end{gather}
 It is clear that
 \[
 Y_{\tau}(t_k)=Y^k.
 \]

 By It\^o's calculus, it is easily found that,
 \begin{multline}\label{eq:dissipationrela}
 \frac{d}{dt}\mathbb{E}|Y(t)-Y_{\tau}(t)|^2
 =2\mathbb{E}(Y(t)-Y_{\tau}(t))\cdot
 \Big\{[g_1(Y(t))-g_1(Y_{\tau}(t_k))] \\
 +g_2(Y(t))+\frac{1}{p-1}\sum_{j\in \bm{\xi}_k}\nabla f_k(Y_{\tau}(t_k)) \Big\}.
 \end{multline}

 Using \eqref{eq:auxsde1} and \eqref{eq:auxsde2},
 \begin{multline*}
 Y(t)-Y_{\tau}(t)=Y(t_k)-Y_{\tau}(t_k)+\int_{t_k}^t[g_1(Y(s))-g_1(Y_{\tau}(t_k))]\,ds \\
 +\int_{t_k}^t \left[g_2(Y(s))+\frac{1}{p-1}\sum_{j\in \bm{\xi}_k}\nabla f_k(Y_{\tau}(t_k)) \right]\,ds.
 \end{multline*}
 Using this formula, we can decompose the right hand side of \eqref{eq:dissipationrela} as $I_1+I_2$,
 where
 \begin{multline}
 I_1=2\mathbb{E}\Bigg[Y(t_k)-Y_{\tau}(t_k)+\int_{t_k}^t[g_1(Y(s))-g_1(Y_{\tau}(t_k))]\,ds \Bigg]
 \cdot   \\\Big[(g_1(Y(t))-g_1(Y_{\tau}(t_k)))
+ g_2(Y(t))+\frac{1}{p-1}\sum_{j\in \bm{\xi}_k}\nabla f_k(Y_{\tau}(t_k)) \Big]\\
 +2\mathbb{E}\int_{t_k}^t \left[g_2(Y(s))+\frac{1}{p-1}\sum_{j\in \bm{\xi}_k}\nabla f_k(Y_{\tau}(t_k)) \right]\,ds\cdot \big[g_1(Y(t))-g_1(Y_{\tau}(t_k))\big],
 \end{multline}
 and
  \begin{multline}
 I_2=\mathbb{E}\int_{t_k}^t \Big[g_2(Y(s))+\frac{1}{p-1}\sum_{j\in \bm{\xi}_k}\nabla f_k(Y_{\tau}(t_k)) \Big]  \cdot \Big[g_2(Y(t))+\frac{1}{p-1}\sum_{j\in \bm{\xi}_k}\nabla f_k(Y_{\tau}(t_k)) \Big]\,ds.
 \end{multline}

 Since $Y(t)$ and $Y_{\tau}(t_k)$ are independent of random batch at $t_k$ (namely $\bm{\xi}_k$), one has by Lemma \ref{lmm:randomforce} that
\begin{multline}
 I_1=2\mathbb{E}\Bigg\{\Big[Y(t_k)-Y_{\tau}(t_k)+\int_{t_k}^t(g_1(Y(s))-g_1(Y_{\tau}(t_k)))\,ds \Big]
 \cdot   \\\Big[g_1(Y(t))+g_2(Y(t))- g_1(Y_{\tau}(t_k))-g_2(Y_{\tau}(t_k)) \Big]\Bigg\}\\
 +2\mathbb{E}\int_{t_k}^t \big[g_2(Y(s))-g_2(Y_{\tau}(t_k)) \big]\cdot \big[g_1(Y(t))-g_1(Y_{\tau}(t_k))\big]\,ds,
 \end{multline}

 In $I_1$, we need to estimate two type of terms. The first type is like
 \[
 J_1:=\mathbb{E}\big[Y(t_k)-Y_{\tau}(t_k)\big]\cdot\big[g_\ell(Y(t))-g_\ell(Y(t_k))\big],
 \]
 while the second type is like
 \[
 J_2:=\mathbb{E}\int_{t_k}^t \big[g_\ell(Y(s))-g_\ell(Y_{\tau}(t_k)) \big]\cdot \big[g_m(Y(t))-g_m(Y_{\tau}(t_k))\big]\,ds,
 \]
where $\ell,m=1$ or $2$.
Note that
 \begin{gather}
 dY= (g_1+g_2)(Y(t))\,dt
 +\sqrt{2 D} dW(t).
 \end{gather}
By It\^o's formula, one has
 \begin{gather}\label{eq:expandtheforce}
 g_\ell(Y(t))-g_\ell(Y(t_k))=g_\ell(Y(t_k))-g_\ell(Y(t_k))
 +\int_{t_k}^t dY\cdot \nabla g_\ell(Y(s))+\int_{t_k}^t \Delta g_\ell(Y(s)) D ds.
 \end{gather}
Hence,
 \[
 J_1\le C\|Y(t_k)-Y_{\tau}(t_k)\|^2+C\|Y(t_k)-Y_{\tau}(t_k)\|(1+D)\tau,
 \]
 where we used the fact that
 \[
 \mathbb{E}\big[Y(t_k)-Y_{\tau}(t_k)\big]\cdot \int_{t_k}^t dY\cdot \nabla g_\ell(Y(s))
 =\mathbb{E}\big[Y(t_k)-Y_{\tau}(t_k)\big]\cdot \int_{t_k}^t (g_1+g_2)\cdot \nabla g_\ell(Y(s))\,ds.
 \]
Using similar expansions, the terms like $J_2$ can be controlled as
 \[
 J_2\le C\|Y(t_k)-Y_{\tau}(t_k)\|^2\tau+C\|Y(t_k)-Y_{\tau}(t_k)\|(1+D)\tau^2+CD\tau^2+C(1+D^2)\tau^3.
 \]
Overall, we have,
 \[
 \begin{split}
 I_1 &\le C\|Y(t_k)-Y_{\tau}(t_k)\|^2+C\|Y(t_k)-Y_{\tau}(t_k)\|(1+D)\tau+CD\tau^2+C(1+D^2)\tau^3\\
 & \le C\|Y(t_k)-Y_{\tau}(t_k)\|^2+C(1+D^2)\tau^2.
 \end{split}
 \]

We now proceed the estimate of $I_2$.
 Using again \eqref{eq:expandtheforce},
 one obtains
 \begin{multline*}
 I_2\le \mathbb{E}\left|-\frac{1}{N-1}\sum_{j=1}^{N-1}\nabla f_j(Y(t_k))+\frac{1}{p-1}\sum_{j\in \bm{\xi}_k}\nabla f_k(Y_{\tau}(t_k)) \right|^2 \tau\\
 +C\left\|-\frac{1}{N-1}\sum_{j=1}^{N-1}\nabla f_j(Y(t_k))+\frac{1}{p-1}\sum_{j\in \bm{\xi}_k}\nabla f_k(Y_{\tau}(t_k))\right\|(1+D)\tau^2\\
+CD\tau^2+C(1+D^2)\tau^3.
 \end{multline*}
The first term is nothing but the variance of the random force, and thus Lemma \ref{lmm:randomforce} gives
 \[
 I_2\le 2\left(\frac{1}{p-1}-\frac{1}{N-1}\right)\|\Lambda\|_{\infty}\tau +CD\tau^2+C(1+D^2)\tau^3.
 \]
Hence, we eventually have
 \begin{multline*}
 \|Y(t_{k+1})-Y_{\tau}(t_{k+1})\|^2
 = \|Y(t_{k})-Y_{\tau}(t_{k})\|^2
 +\int_{t_k}^{t_{k+1}}(I_1(s)+I_2(s))\,ds \\
 \le (1+C\tau)\|Y(t_{k})-Y_{\tau}(t_{k})\|^2
+2\left(\frac{1}{p-1}-\frac{1}{N-1}\right)\|\Lambda\|_{\infty}\tau^2+C(1+D^2)\tau^2.
 \end{multline*}
 This implies that
 \[
 \begin{split}
 \|Y(t_{k})-Y_{\tau}(t_{k})\|^2 &\lesssim (e^{Ck\tau }-1)\left(\frac{1}{p-1}\|\Lambda\|_{\infty}
 \tau+(1+D^2)\tau^2\right)\\
 &\lesssim (e^{Ck\tau }-1)\left[\frac{1}{p-1}\|\nabla U_1\|_{\infty}
 \tau+(1+D^2)\tau^2\right].
 \end{split}
 \]
 By the definition of $W_2$ distance \eqref{eq:W2}, the claim follows.
\end{proof}

\subsection{The algorithm and some comments }

We now combine the ideas in section \ref{sec:split} and section \ref{sec:sderandombatch} into a practical algorithm, called the Random-batch Monte Carlo (RBMC) method (detailed in Algorithm \ref{fastmcmc}). In particular, we use the splitting strategy and decompose $U$ into a smooth part with long range and a singular part with short range.
For the smooth part, we apply the mini-batch idea for the SDE steps similar to SGLD while we use the singular part for the rejection.

\begin{algorithm}[H]
\caption{(Random-batch Monte Carlo algorithm)}\label{fastmcmc}
\begin{algorithmic}[1]
\State Split $U:=U_1+U_2$ such that $U_1$ is smooth and with long range and $U_2$ is with short range.
Randomly generate $N$ initial particles, set $N_s$ be the total number of samples and choose $p>1$, $m\ge1$
\For {$n \text{ in } 1: N_s$}
\State Randomly choose a particle $i\in \{1, \cdots, N\}$ with uniform probability
       \State $X^i \leftarrow x_i$
     \For {$k=1,\cdots, m$}
      \State  Choose $\bm{\xi}_k$, $z_k\sim \mathcal{N}(0, I_d)$, $\tau_k>0$ and let,
 \begin{gather*}
 X^i\leftarrow X^i-\tau_k \left[\frac{\nabla V(X^i)}{w(N-1)}+\frac{1}{p-1}\sum_{j\in \bm{\xi}_k} \nabla U_1(X^i-x_j)\right]+\sqrt{\frac{2\tau_k}{(N-1)w^2 \beta}} z_k
 \end{gather*}
      \EndFor
      \State Let $x_i^*\leftarrow X^i$. Compute the following using cell list or other data structures:
      \[
      \alpha=\min\left\{1, \exp\Big[-\beta\sum_{j: j\neq i}  w^2(U_2(x_i^*-x_j)-U_2(x_i-x_j))\Big]  \right\}
      \]
      \State Generate a random number $\zeta$ from uniform distribution on $[0, 1]$. If $\zeta\le \alpha$, set
      \[
      x_i \leftarrow x_i^*
      \]
\EndFor
\end{algorithmic}
\end{algorithm}

Compared to the traditional Metropolis algorithm \cite{hastings1970monte}, our new algorithm \ref{fastmcmc} has the following advantages:
\begin{itemize}
\item The computational cost is $O(1)$ for each iteration;
\item There is no rejection in {\it Step 2}. The efficiency could be higher.
\end{itemize}

When we run the implementations, there are two phases for the algorithms. Phase 1 is the burn-in phase, in which the distribution of the Markov chain converges the desired Gibbs distribution. Phase 2 is the sampling phase. We run the algorithms for many iterations and collect these $N$ points for selected iterations as the sampling points.  As the iteration goes on, the number of sampling points gets larger and larger.

Below, we perform some discussions.
\begin{enumerate}
\item
For the molecular regime, where $w=1$.  Taking $\gamma=1$ in \eqref{eq:langevinsampling}
and using the same splitting, one may have the following SDE
\begin{gather}\label{eq:SDEaux1}
dX^i=-\sum_{j:j\neq i}\nabla U_1(X^i-x_j)\,dt+\sqrt{\frac{2}{\beta}}\,dW^i.
\end{gather}
If one uses mini-batch idea,  the random force then becomes
\[
-\frac{N-1}{p-1}\sum_{j\in \bm{\xi}_k} \nabla U_1(X^i-x_j).
\]
For this type of forcing term, the magnitude can be as large as $O(N)$, and one needs $\tau_k\sim O(1/N)$ for the numerical methods to be stable. This is not preferred in practice.  Taking $\gamma=1/(N-1)$ corresponds to nothing but rescaling the time with $\tilde{t}=(N-1)t$.

\item
 If we move $N$ particles, then {\it Step 2} can be solved using the RBM in \cite{jin2018random} to reduce the cost. However, moving $N$ particles will make the following event to have high probability
 \[
 A=\{\exists (i, j) \text{~such that } |x_i^*-x_j^*|\le \delta*\text{mean distance}\}
 \]
 This event will be rejected with high probability in {\it Step 3}. Hence, moving $N$ particles is not preferred in practice.
\end{enumerate}

\section{Numerical examples}\label{sec:numerical}

In this section, we test our RBMC sampling algorithm on two examples, and compare to the Metropolis-Hastings algorithm (see Algorithm \ref{alg:mh}), which is a classical MCMC method. The first example is the interacting particle system from the Dyson Brownian motion with the Gibbs measure to be the semicircle law in the $N\to\infty$ limit. The second example is the Lennard-Jones fluids and we aim to recover the equation of states $P=P(\rho)$, where $P$ is the pressure and $\rho$ is the density. Both examples have singular interaction potentials. The difference is that the first example has a long range while the second example has a short range.  Even for short range potentials, our algorithm can still be applied.

Below, an ``iteration'' will refer to one complete move of a particle, including the rejection step. In other words, this will correspond to the loop of $n$ in Algorithm \ref{fastmcmc} and Algorithm \ref{alg:mh}.
All the numerical experiments are implemented using C++ and performed in a Windows system with  Intel Core i7-6700HQ CPU @ 2.60GHz and  8 GB memory.

\begin{algorithm}[H]
	\caption{(Metropolis-Hastings algorithm )}\label{alg:mh}
	\begin{algorithmic}[1]
		\For {$n \text{ in } 1: N_s$}
		\State Randomly choose a particle $i\in \{1, \cdots, N\}$ with uniform probability
		\State Move $x_i$ with some random movement and obtain $x_i^*$
		\State Compute the following using cell list or other data structures:
		\[
		\alpha=\min\left\{1, \exp\Big[-\beta w(V(x_i^*)-V(x_i))-\beta\sum_{j: j\neq i}    w^2(U(x_i^*-x_j)-U(x_i-x_j))\Big]  \right\}
		\]
		\State Generate $\zeta\sim U[0, 1]$. If $\zeta\le \alpha$, set
		\[
		x_i\leftarrow x_i^*
		\]
		\EndFor
	\end{algorithmic}
\end{algorithm}

\subsection{Dyson Brownian motion}\label{subsec:db}
To test our method, we first focus on an interacting particle system from Dyson Brownian motion arising from random matrix theory \cite{tao2012,erdos2017}. Consider a Hermitian matrix valued Ornstein-Uhlenbeck process
\begin{gather}
dA=-A\,dt+\frac{1}{\sqrt{N}}dB,
\end{gather}
where the matrix $B$ is a Hermitian matrix whose diagonal elements are independent standard Brownian motions while the off-diagonal elements in the upper triangular half are of the form $(B_R+i B_I)/\sqrt{2}$, with $B_R$ and $B_I$ being  independent standard Brownian motions. It can be shown \cite{tao2012,erdos2017} that the eigenvalues of $A$ satisfy the following  system of SDEs ($1\le j\le N$), called the Dyson Brownian motion:
\begin{gather}
d\lambda_j(t) =-\lambda_j(t)\,dt+\frac{1}{N}\sum_{k: k\neq j}\frac{1}{\lambda_j-\lambda_k}dt
+\frac{1}{\sqrt{N}} dW_j,
\end{gather}
where $\{W_j\}$'s are independent standard Brownian motions and we use $W_j$ instead of $B_j$  to be consistent with the notations throughout this paper. The density of these eigenvalues in the limit $N\to\infty$ is given by
\begin{gather}\label{eq:dysonlimiteq}
\partial_t\rho(x,t)+\partial_x(\rho(u-x))=0, ~~u(x, t)=\pi(H\rho)(x, t)=\mathrm{p.v.}\int_{\mathbb{R}}\frac{\rho(y,t)}{x-y}\,dy,
\end{gather}
where $H(\cdot)$ is the Hilbert transform on $\mathbb{R}$, $\pi=3.14\cdots$ is the circumference ratio and p.v. means the integral is evaluated using the Cauchy principal value.
The corresponding limiting equation \eqref{eq:dysonlimiteq} has an invariant measure, given by the semicircle law:
\begin{gather}\label{eq:semicircle}
\rho(x)=\frac{1}{\pi}\sqrt{2-x^2}.
\end{gather}
Note that the mean field limit \eqref{eq:dysonlimiteq} is the same for any SDE of the following form
 \begin{gather}\label{eq:dysonBMips}
d\lambda_j(t) =-\lambda_j(t)\,dt+\frac{1}{N-1}\sum_{k: k\neq j}\frac{1}{\lambda_j-\lambda_k}dt
+\sqrt{\frac{2\mu}{N-1}} dW_j.
\end{gather}
We changed the prefactor $1/N$ into $1/(N-1)$ for convenience.

For the system \eqref{eq:dysonBMips}, clearly we have
$w=1/(N-1)$ and the physical temperature in this case is given by
$T=\mu/(N-1)^2$, and thus $\beta=(N-1)^2/\mu$. The corresponding Gibbs distribution is
\[
\pi \propto \exp\left[-\beta\Big(\frac{1}{2(N-1)}\sum_ix_i^2-\frac{1}{(N-1)^2}\sum_{i,j: i<j}\ln|x_i-x_j|\Big)\right].
\]
The binary potential is given by
\begin{gather}
U(x_i-x_j)=u(|x_i-x_j|)=-\ln|x_i-x_j|.
\end{gather}
We consider $N=500$ particles sampled from the uniform distribution on $[-5,5]$ and run Algorithm \ref{fastmcmc}.
Let $r=|x_i-x_j|.$
We consider the splitting $U=U_1+U_2$ with $U_{1,2}(x)=u_{1,2}(|r|)$ where
\begin{eqnarray}\label{eq:dysonsplit1}
u_1(r)=
\begin{cases}
\ln(100)-100r+1    &0<r<0.01,\\
-\ln(r)    &r\geq 0.01,
\end{cases}
\end{eqnarray}
and
\begin{eqnarray}\label{eq:dysonsplit2}
u_2(r)=
\begin{cases}
-\ln(100r)+100r-1    &0<r<0.01,\\
0    &r\geq 0.01.
\end{cases}
\end{eqnarray}
Clearly, $U_1(x_i-x_j)=u_1(r)$ is the interaction extending to $r=0$ along the tangent line at $r=0.01$.
The short-range potential $u_2(r)$ is nonzero only if $r<0.01$, so we use cell list to calculate $U_2(x_i-x_j)$ in the nearby boxes of a particular particle with length $0.01$.  The reason to choose $0.01$ as the splitting point is that we hope that there are only a few particles in the range of $U_2$ for a chosen particle.

We fix the parameter $\mu=1$. The batch and step sizes $p=2$ and $m=9$ are taken so that
the SDE step is given by
\begin{gather}
X^i \leftarrow X^i-\big[X^i+\nabla U_1(X^i-x_{\eta_k})\big]\tau+\sqrt{\dfrac{2\mu}{N-1}}\sqrt{\tau}z_k,
~~k=0,\cdots, m-1,
\end{gather}
where $\eta_k$ is chosen from $\{1, 2, \cdots, N\}\setminus\{i\}$ with uniform probability. The time step is chosen as
$
\tau= 10^{-4}
$
for all the experiments in this example. This small time step is needed largely due to the large magnitude of the gradient for $U_1$ in $r<0.01$ region.

\begin{table}
       \centering
	\begin{tabular}{|c|c|c|}
		\hline
		&Iteration& Time (s)\\\hline
		MH&3e5&36.8 \\\hline
		RBMC&3e6&14.5 \\\hline
	\end{tabular}
\caption{Burn-in phase of the two methods. First 3e5 iterations in MH method are discarded, which costs $36.8$s. First 3e6 iterations are discarded in RBMC, which costs $14.5$s.}
\label{tabl:burninDB}
\end{table}

As we have mentioned, there are two phases: the burn-in phase and sampling phase. Table \ref{tabl:burninDB} shows the burn-in phase we pick for the two methods in one experiment. The first $3\times 10^5$ sampling iterations in the MH method and the first $3\times 10^6$ sampling iterations in the RBMC method are discarded as the distribution is far away from the semicircle law. Though the number of iterations is larger, our proposed method costs only $14.5s$ compared to $36.8s$ in MH, implying that our method is more efficient.

\begin{figure}[ht]	
	\centering
	\includegraphics[width=0.6\textwidth]{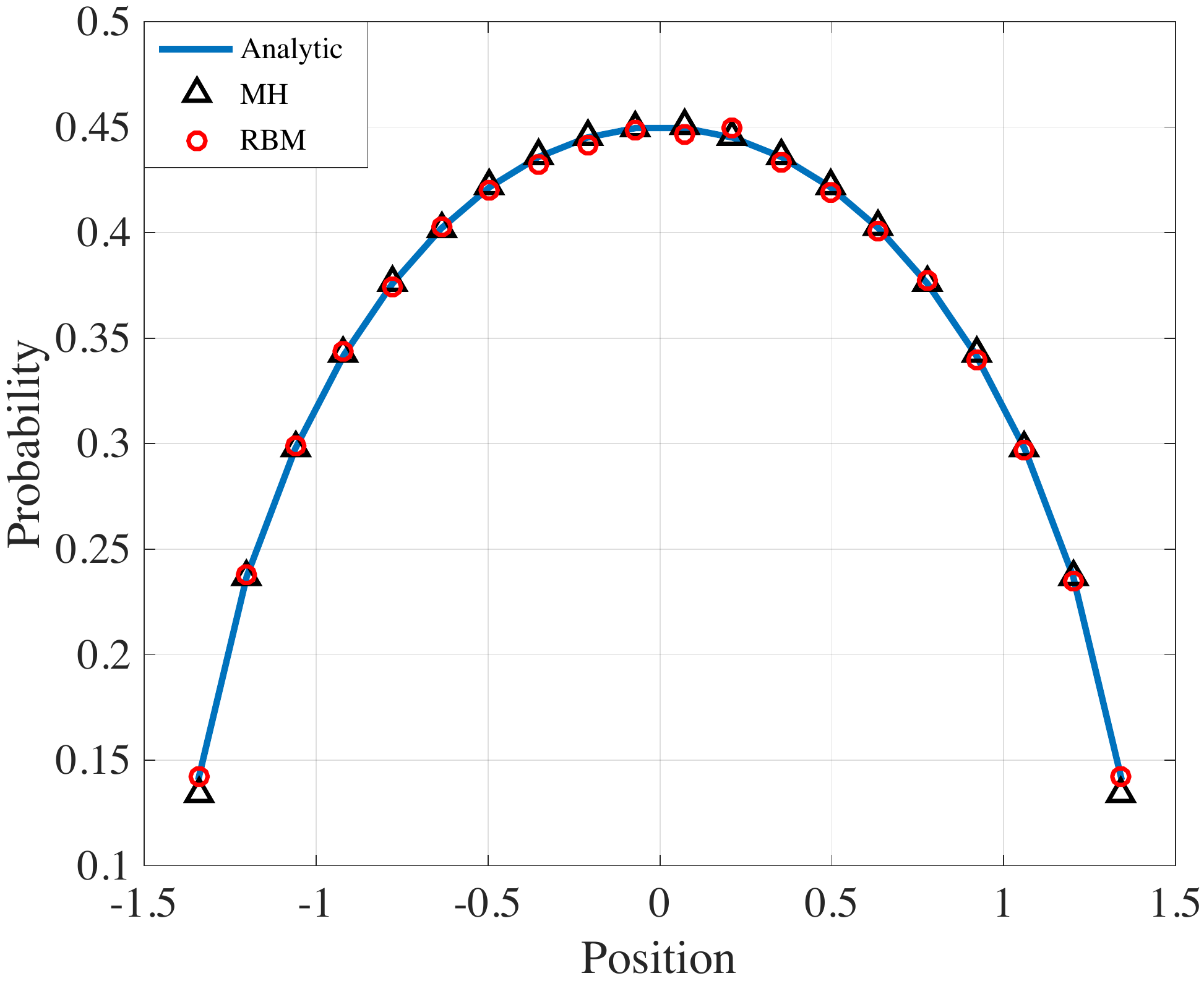}
	\caption{Empirical densities obtained by the two methods with 1e7 sampling iterations (1e7$N$ sample points). The blue curve represents the semicircle law \eqref{eq:semicircle}.}
	\label{fig:dyson}
\end{figure}

In Fig. \ref{fig:dyson}, we show the performances of the two methods, with the samples collected from the sampling phase. The empirical densities are computed using bin counting. In particular, we divide the interval $[-\sqrt{2}, \sqrt{2}]$ into $20$ equal pieces so that the width of each bin is $h=\sqrt{2}/10$. Define $x_i=-\sqrt{2}+ih$. Then, the density in the $i$th bin $[x_{i-1}, x_{i})$ is approximated by
$
\rho_i\approx  N_i/(N h).
$
The samples are collected from all the iterations after the burn-in phase (each iteration contributes $N=500$ sample points).
In particular, let $N_i$ denote the number of particles in the bin $[x_{i-1}, x_{i})$. In each iteration $k$, we do $N_i\leftarrow N_i+N_i^k$, where $N_i^k$ is the number of points in the $i$th bin at the $k$th iteration. At last, normalize $N_i$: $N_i=N_i/\sum N_j$. Clearly, Fig. \ref{fig:dyson} shows that both methods capture the semicircle law \eqref{eq:semicircle}.

\begin{table}
	\begin{tabular}{|c|c|c|c|c|c|c|c|c|}
		\hline
		Iteration&1e5&2e5&5e5&1e6&2e6&5e6&1e7&2e7\\\hline
		Time(MH)(s)&20.5&36.3&88.2&188.3&314.9&750.1&1384.9&2952.7 \\\hline
		Time(RBMC)(s)&2.4&2.5&3.8&4.6&13.2&30.7&62.0&126.5 \\\hline
		Error(MH)&0.035&0.017&0.0060&0.0038&0.0023&0.0016&0.0015&0.0014 \\\hline
		Error(RBMC)&0.012&0.011&0.0062&0.0051&0.0048&0.0031&0.0028&0.0022	\\\hline
	\end{tabular}
\caption{Sampling times and errors for two methods. }
\label{tabl:dstime}
\end{table}

To compare the performance of the algorithms in more details, in Table \ref{tabl:dstime}, we compare the running times (spent in the sampling phase only) and relative errors of the two methods as the iteration goes on.  Here, the iteration means the number of iterations in the sampling phase and the relative error is computed by the $L^1$ norm
\[
E:=\sum_{i}\left|\rho_i-h^{-1}\int_{x_{i-1}}^{x_i}\rho(x)\,dx\right| h \Bigg/\int_{-\sqrt{2}}^{\sqrt{2}} \rho(x)\,dx
=\sum_{i}\left|\rho_i-h^{-1}\int_{x_{i-1}}^{x_i}\rho(x)\,dx\right| h.
\]
We also plot the running time versus the Monte Carlo errors in Fig. \ref{fig:DBtime-err}. The data in time-error plots in our figure are collected at different iterations in a single experiment with fixed $N$. A more reasonable comparison would be based on the data collected from different $N$ values and the time should be for the distribution to reach the equilibrium. This will be done in our subsequent work for the interactions with long ranges like the Coulomb interaction.

As can be seen in Table \ref{tabl:dstime} and Fig. \ref{fig:DBtime-err}, our proposed method is more efficient if the desired
accuracy is not very high. The reason is that our proposed method takes only $O(1)$ time per iteration.
The overall running time in the RBMC method is roughly $1/20$ of that in the MH method for obtaining the same number of samples. As can be seen in Fig. \ref{fig:DBtime-err}, to achieve the same sampling error, the computation time needed in the RBMC method is like  $1/10$ of that needed in the MH method.

\begin{figure}[ht]	
	\centering
	\includegraphics[width=0.6\textwidth]{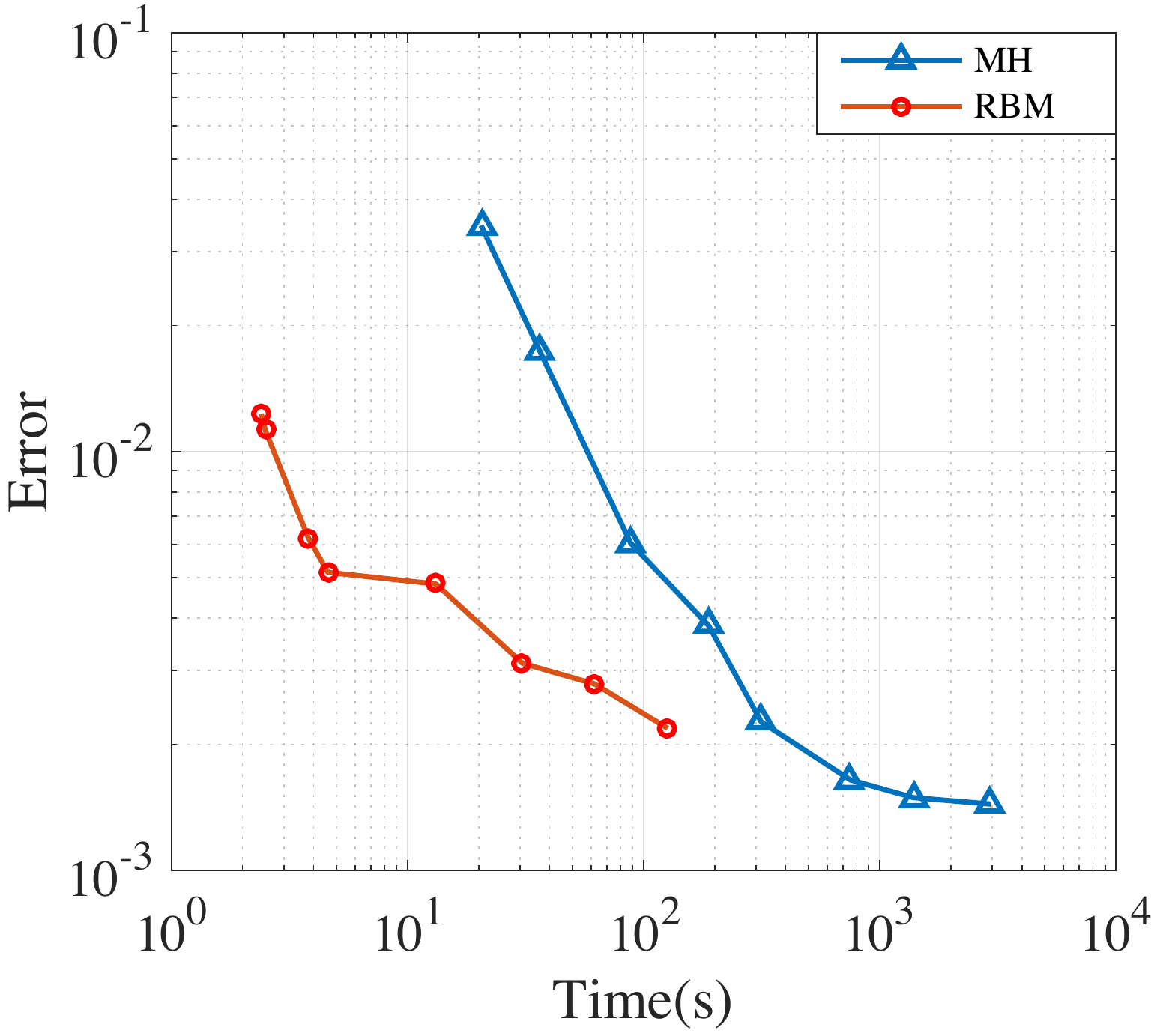}
	\caption{Relative error of two methods decreases as time goes on. In RBMC method, we only need $10\%$ time of MH method to get the error with the same level.}
	\label{fig:DBtime-err}
\end{figure}

\subsection{Lennard-Jones fluids}\label{subsec:ljnumerics}

In this subsection, we test our algorithm on the Lennard-Jones fluid, aiming to recover the equation of state,
which has been studied in \cite{johnson1993lennard}.
We compare the sampling performance of the standard MH algorithm and our proposed method.

For the Lennard-Jones fluids, we take the external potential $V\equiv 0$, and consider the regime for the molecular dynamics so that $w=1$. Hence, the Gibbs measure is given by
\[
\pi \propto \exp\left[-\beta \sum_{i,j: i<j}U(x_i-x_j)\right].
\]
The potential for this model system is given by:
\begin{align}\label{eq:ljpotential}
U(x_i-x_j)=u(r_{ij})=4\left[\left(\dfrac{1}{r_{ij}}\right)^{12}-\left(\dfrac{1}{r_{ij}}\right)^6\right]
\end{align}
where again
$r_{ij}=|x_i-x_j|$
is the distance between three-dimensional locations $x_i$ and $x_j$. For the convenience, we will also define
$
\vec{r}_{ij}:=x_i-x_j.
$

In numerical simulations, the periodic boxes are often used to approximate the fluid. Assume the length of the periodic box is $L$.  A given particle $x_i$ interacts with all other particles in the system and their periodic images.  We pick a cutoff length
$r_c=L/2.$
Thanks to the fast decay property of the Lennard-Jones potential \eqref{eq:ljpotential}, for a fixed particle, one can approximate the summation away beyond $r_c$ by continuous integral, and the radial density function (see \cite[Chap. 3]{frenkel2001understanding}) is approximated by
$\rho(r)\approx N/L^3$ when $r\ge r_c$.

With the approximation, the pressure formula is given by
\begin{align}\label{eq:pressureforperiodic}
P=\dfrac{\rho}{\beta}+\frac{8}{V}\sum_{i=1}^N \sum\limits_{j: j>i,\tilde{r}_{ij}<r_c}(2\tilde{r}_{ij}^{-12}-\tilde{r}_{ij}^{-6})+\dfrac{16}{3}\pi\rho^2 \left[\dfrac{2}{3}\left(\dfrac{1}{r_c}\right)^9-\left(\dfrac{1}{r_c}\right)^3\right],
\end{align}
and  the acceptance rate in the MH algorithm (Algorithm \ref{alg:mh}) can be approximated by
\begin{gather}\label{eq:mhrate}
\alpha=\min\left\{1, \exp\left[-\beta\Big( \sum_{j: j\neq i, \tilde{r}_{ij}^*< r_c}   u(\tilde{r}_{ij}^*)-\sum_{j: j\neq i, \tilde{r}_{ij}< r_c} u(\tilde{r}_{ij})\Big)\right]  \right\}.
\end{gather}
Here, we have introduced
\[
\tilde{r}_{ij}=|\vec{r}_{ij}+\vec{n}L|
\]
for some suitably three-dimensional integer vector $\vec{n}$ so that
$|\vec{r}_{ij}+\vec{n}L|$ is minimized. In fact, since $r_c=L/2$, for each $j$, there is at most one image (including itself) that falls into $B(x_i, r_c)$, so the formula is fine.  Clearly,  $\tilde{r}_{ij}^*$ is similarly defined just with $x_i$ replaced by $x_i^*$, the position of the candidate move. Regarding Algorithm \ref{fastmcmc}, for a particle, under the same approximation, the net forces of all the particles/images away beyond $r_c$ contribute net zero force to the current particle.
Equivalently, we think that each of these images contributes zero force to the current particle. Hence, in the SDE step of Algorithm \ref{fastmcmc}, only when the randomly selected particle/image (select $j\in \{1,\cdots, N\}\setminus \{i\}$ randomly and then determine the $\vec{n}$ vector such that $|x_i-x_j+\vec{n}L|$ is minimized) is within $B(X^i, r_c)$ will we compute its contribution and update $X^i$; otherwise, we do not update $X^i$ but still increase $k$ until it reaches $m$. For the rejection step, since $U_2$ is short-ranged, we also only consider images within $B(X^i, r_c)$.
We refer the readers to Appendix \ref{app:pressureandenergyLJ} for the details of these formulas.

Following  \cite{frenkel2001understanding}, we take the number of particles $N=500$
and the length of the periodic boxes will be determined by the desired density correspondingly
\begin{gather}
L=\left(\frac{N}{\rho} \right)^{1/3}.
\end{gather}
For the MH method (Algorithm \ref{alg:mh}), the random displacement for a randomly selected particle obeys the
$\mathcal{N}(0, \sigma^2 I_3)$ with $\sigma=0.05 r_c.$
In other words,
\begin{gather}
x_i^{\ast}= x_i+\sigma z,~z\sim \mathcal{N}(0, I_3).
\end{gather}

For our proposed algorithm (Algorithm \ref{fastmcmc}), we do splitting
\begin{gather}
u(r)=u_1(r)+u_2(r)
\end{gather}
where
\begin{eqnarray}
u_1(r)=
\begin{cases}
	2^{-1/3}(r-\sqrt[6]{2})^2-1,   &0<r<\sqrt[6]{2},\\
	4\times(r^{-12}-r^{-6}) ,   &r\geq \sqrt[6]{2},
\end{cases}
\end{eqnarray}
and
\begin{eqnarray}
u_2(r)=
\begin{cases}
4\times(r^{-12}-r^{-6})-2^{-1/3}(r-\sqrt[6]{2})^2+1 ,   &0<r<\sqrt[6]{2},\\
0,    &r\geq \sqrt[6]{2}.
\end{cases}
\end{eqnarray}
In the SDE step, the batch size is chosen as $p=2$,
and we take $m=9$ so that the iteration in the SDE step is given by
\begin{align}
X^i \leftarrow X^i-\nabla U_1(X^i-x_{\eta_k})\tau+\sqrt{\dfrac{2T}{N-1}}\sqrt{\tau} z_k,~~k=0,\cdots, m-1,
\end{align}
where $\eta_k$ is chosen from $\{1, 2, \cdots, N\}\setminus\{i\}$ with uniform probability. The time step is chosen as
$
\tau=0.01
$
for all the experiments in this example.

\begin{table}
       \centering
	\begin{tabular}{|c|c|c|}
		\hline
		&Iteration& Time (s)\\\hline
		MH&5e4&182.6 \\\hline
		RBMC&2e5&75 \\\hline
	\end{tabular}
\caption{Burn-in phase. First 5e4 iterations in MH method are disgarded, which costs $182.6$s. First 2e5 iterations are disgarded in RBMC, which costs $75$s.}
\label{tabl:burninLJ}
\end{table}

Table \ref{tabl:burninLJ} shows the burn-in phase we pick for the two methods in one experiment. Clearly, our method is again more efficient to achieve the thermal equilibriums.
In Fig. \ref{fig3}, we show the results for these two methods with samples collected from the sampling phase. Clearly they both have good convergence and ergodicity performance. Fig. \ref{fig3} (a) shows the results for $T=2$, while  Fig. \ref{fig3} (b) shows the results for the lower temperature $T=0.9$. The analytical solution is given by the fitting curves in \cite{johnson1993lennard}.

\begin{figure}[h]
	\centering
	\includegraphics[width=0.9\textwidth]{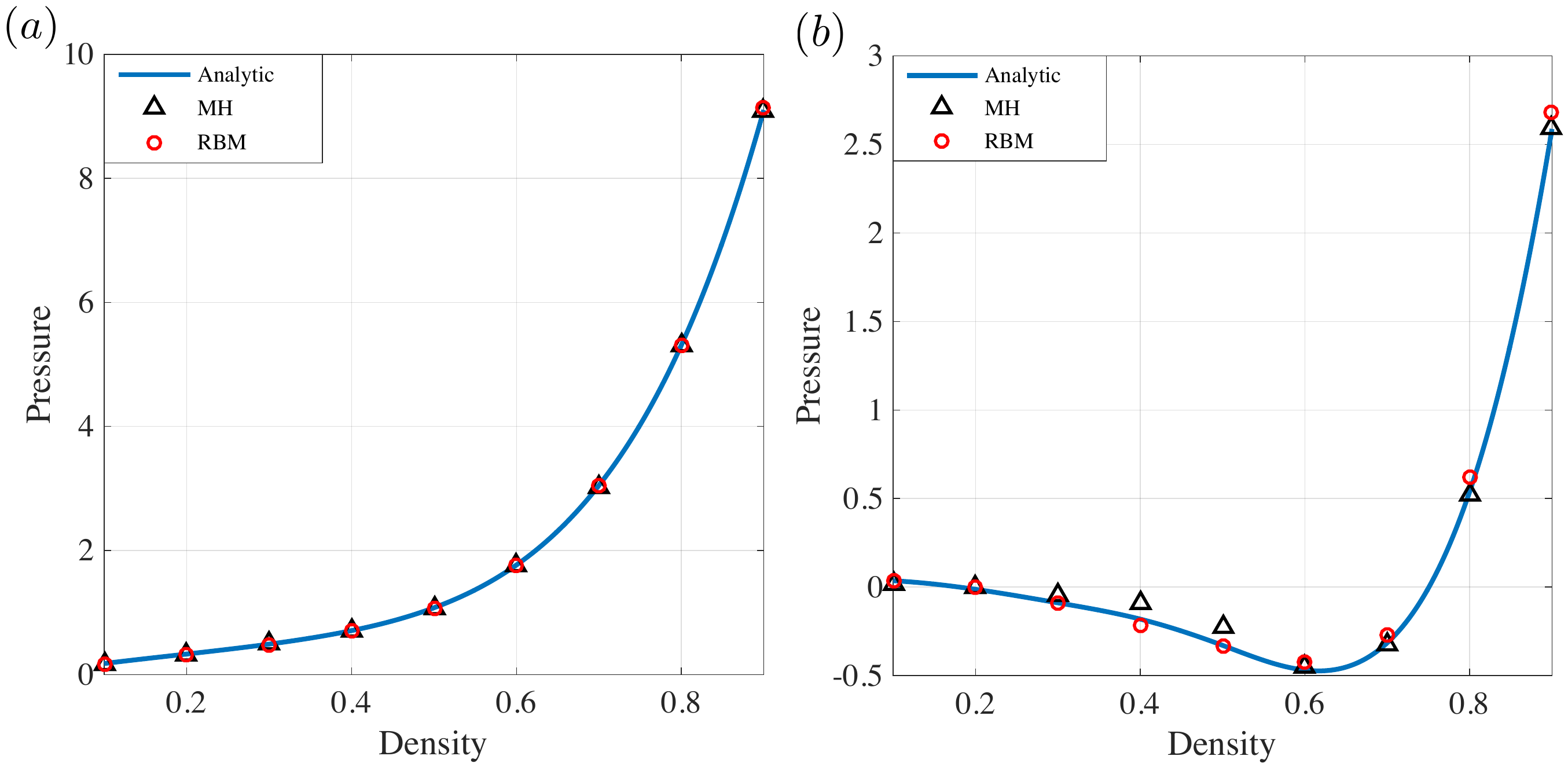}
	\caption{Average results of 2e7 sampling iterations for two methods compared with analytic results. (a) $T=2$; (b) $T=0.9$. }
	\label{fig3}
\end{figure}

For $T=2$, the system stays vapor for all densities while in the case of $T=0.9$, there is
fluid-vapor coexistence for $\rho<0.75$ (when $\rho>0.75$ the system should be in the liquid state).
According to the figure, our proposed algorithms behave similarly to the MH algorithm. They can capture the equation of state correctly.  For $T=2$, the result is quite good and the eventual relative error is like $0.5\%$ and $0.2\%$ relatively as have been mentioned above. For $T=0.9$ though the simulation result is not extremely excellent, the curves have been correctly captured.
The relative errors are about $4\%-5\%$. The reason for the relatively poorer performance when the temperature is low ($T=0.9$) is that the acceptance rates for both methods become smaller. When $\rho=0.4,0.5$, the MH seems to be poorer. In fact, acceptance rates in the MH method are small for these ranges of densities. When $\rho>0.7$ in the RBMC method, the new sampling points generated by the SDE step are more likely to be rejected by the short range potential in Step 2.

\begin{figure}[h]
	\centering
	\includegraphics[width=0.6\textwidth]{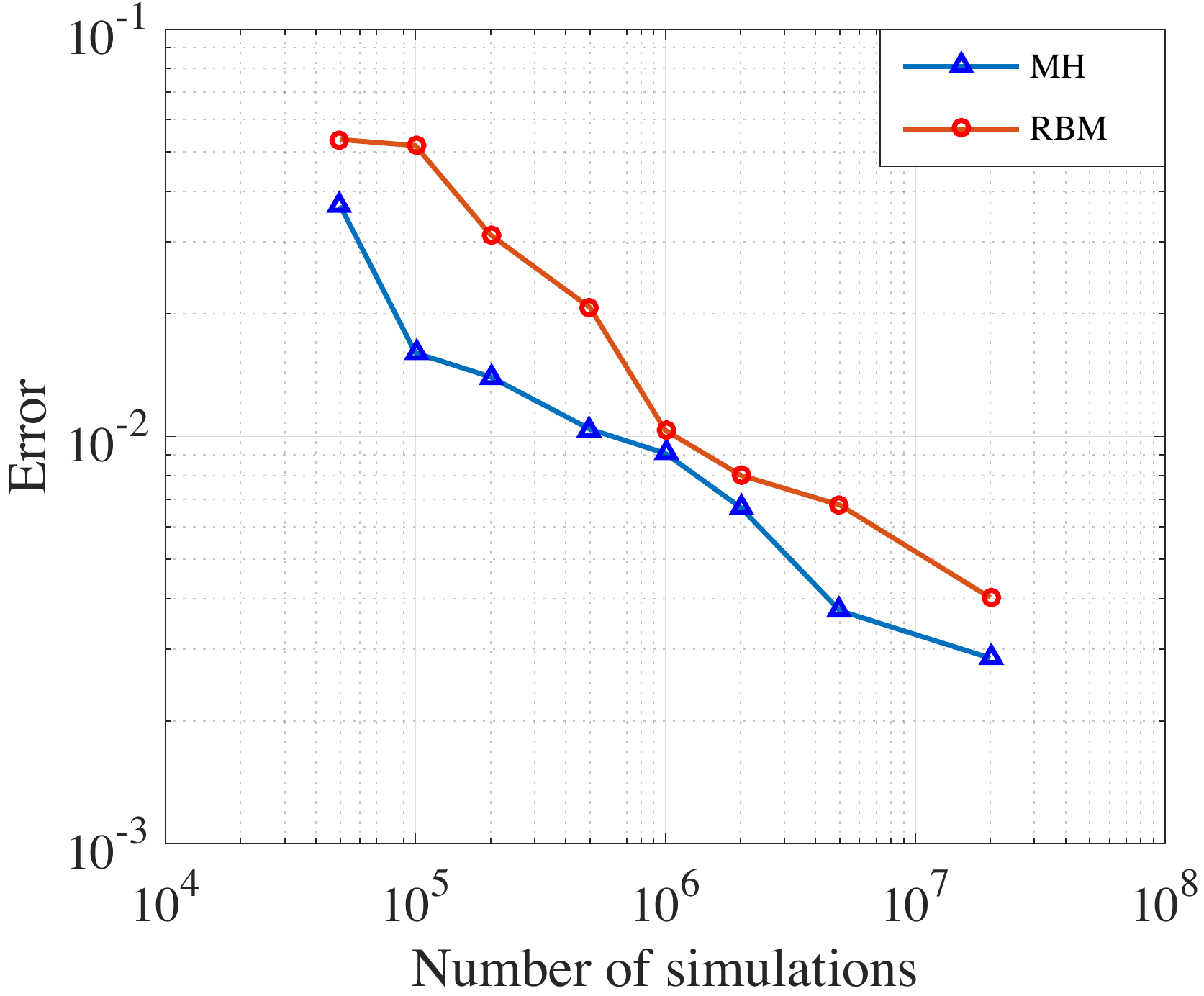}
	\caption{$T=2$. Relative errors for the two methods decrease as the iteration goes on (and thus number of sample points increases). }
	\label{fig4}
\end{figure}

In Fig. \ref{fig4}, we show how the errors change as the iteration goes on in the sampling phase for $T=2$, where the relative error is calculated using $\ell^2$ norm as
\begin{align}
E=\sqrt{\sum\limits_{i=1}\limits^{9}\frac{1}{9}(P_i-P(\rho_i))^2} \Big/\sqrt{\sum_{i=1}^9 \frac{1}{9}P(\rho_i)^2}.
\end{align}
Here, $P_i$ represents the numerical solution while $P(\cdot)$ is the analytical solution given by the fitting curves in \cite{johnson1993lennard}. The samples are also taken from the iterations after the burn-in: in each iteration, we compute the pressure using \eqref{eq:pressureforperiodic} and use all the pressures in the history up to the current iteration to calculate the empirical mean, which is used as the numerical pressure, even though some of them are rejected. The rejected points are necessary for satisfying detailed balance condition.
 According to the picture, the MH method has smaller error than the RBMC method in the same number of simulations. In long run, the relative error of RBMC method is about $0.5\%$, but in the MH method it is about $0.2\%$. The error in two methods decreases as $1/\sqrt{N}$, which satisfies the convergence rate in Monte Carlo method.

\begin{figure}[h]
	\centering
	\includegraphics[width=0.6\textwidth]{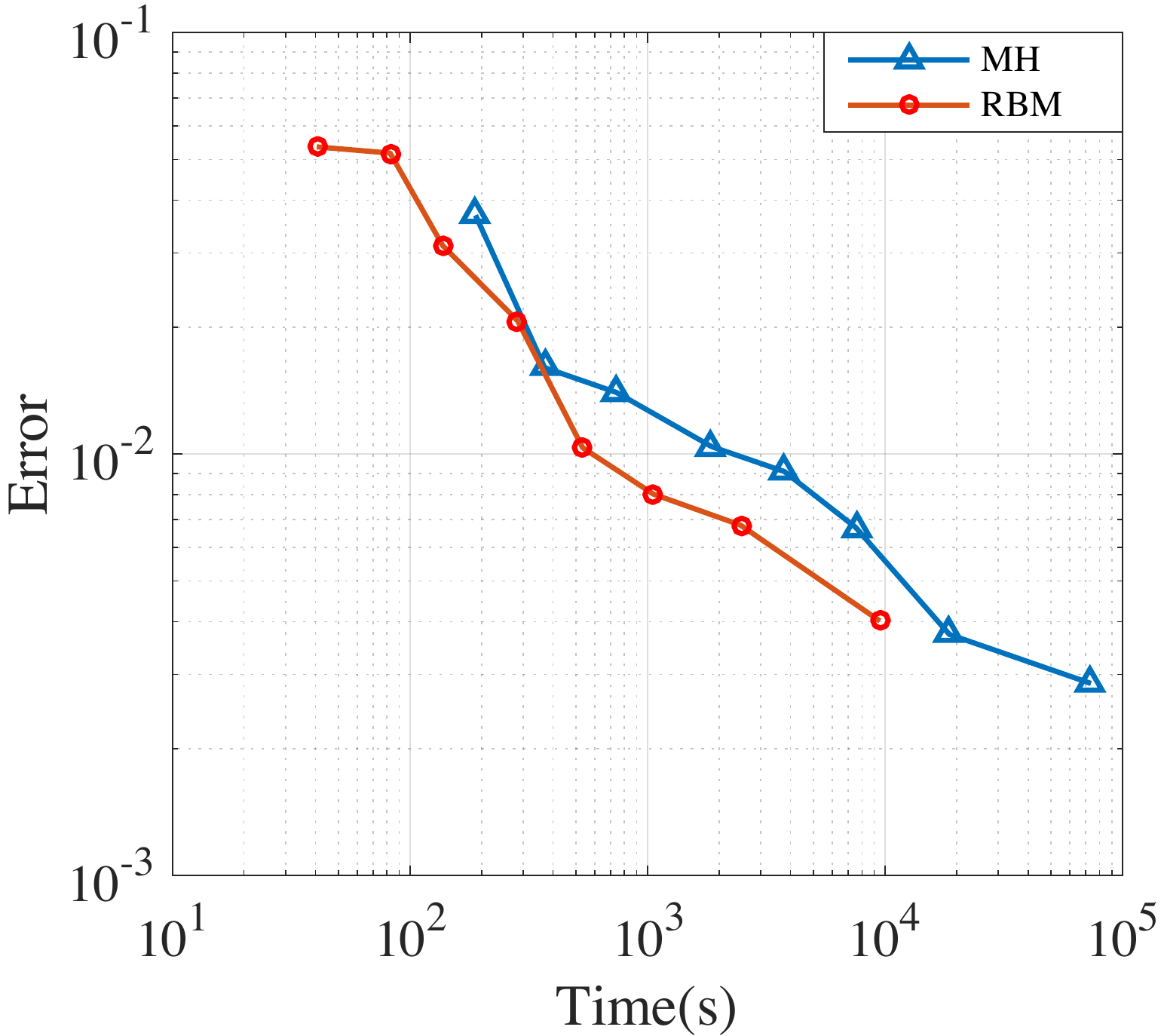}
	\caption{Relative error of densities in $L^1$ norm. For RBMC method, we only need $1/3$ time of MH method to get the accuracy in the same level.}
	\label{fig5}
\end{figure}

\begin{table}
	\begin{tabular}{|c|c|c|c|c|c|c|c|c|}
		\hline
		Iteration&5e4&1e5&2e5&5e5&1e6&2e6&5e6&2e7\\\hline
		Time of MH&188.8&369.0&740.9&1861.0&3767.3&7528.6&18656.4&73557.4
		\\\hline
		Time of RBMC&40.5&82.5&137.3&283.9&533.2&1059.6&2463.0&9509.0
		\\\hline
		Error of MH&0.037&0.016&0.014&0.010&0.0090&0.0067&0.0037&0.0029 \\\hline
		Error of RBMC&0.054&0.052&0.031&0.021&0.010&0.0080&0.0068&0.0040	\\\hline
	\end{tabular}
	\caption{Running times and errors for the two methods in the sampling phase. Compared to the MH,  the RBMC needs twice number of iterations to give comparable sampling errors, but the sampling time is much less, such as 2e6 iterations in the MH method and 5e6 iterations in the RBMC method.}
	\label{tabl:LJtime}
\end{table}

To compare the efficiency of the two methods, we again present the relative errors and running times in Table \ref{tabl:LJtime} and plot the errors versus the running times in Fig. \ref{fig5}. The time here is the time in the sampling phase only while the relative errors are computed using the data collected from the sampling phase only.
Here are some observations, to achieve the same sampling error,  our proposed method roughly needs twice of the iterations compared to the MH method, but takes about $1/3$ time of the MH method. Since our method is an approximation method with systematical error, when the number of iterations gets large (more sample points), the sampling error will not decrease to zero; hence our method is more efficient in the early stage.

It seems that the time saving in this example is not that significant as in the Dyson Brownian motion example. The reason
is that in our implementation of cell list, we check all the 27 neighbor cells which is not quite efficient so our proposed method does not save time as significant as we desire. We expect that when $N$ is larger, our algorithm can be more efficient. As mentioned above in section \ref{subsec:db}, a more reasonable comparison would be based on the data collected from different $N$ values and the time should be computed for the distribution to reach the equilibrium. Since our aim in this paper is to propose an algorithm and we have seen it already owns some advantages in the time efficiency, we choose not to do this careful comparison. The careful and detailed comparison of our algorithm to the MH algorithm will be performed in our subsequent work for the interactions with long ranges like the Coulomb interactions, which is more interesting. (After all, the Lennard-Jones potential still has short range of interaction.)
Besides, during our experiments, we observed that the acceptance rates of the RBMC method are like $75\%-95\%$ while the acceptance rates in the MH are like $20\%-50\%$, which also indicates that our proposed method is more efficient.

It is remarked that when $N=500$, the speedup of our algorithm for the Lennard-Jones potential is comparable to the speedup made by MTS-MC \cite{hetenyi2002multiple}. Since our algorithm is $O(1)$ per iteration, we expect that the CPU time will increase linearly in $N$ and thus the speedup could be larger for larger $N$. MTS-MC has roughly the same behavior, but we have theoretic analysis for our algorithm ($O(1)$ per move). One more thing, as we have pointed out, the data structure for the cell list is not optimized in our implementation, and the speedup can be increased if it is improved.

\section{Conclusions} \label{sec:conc}

In summary, a RBMC method is developed for sampling from Gibbs distributions of interacting particle systems with
singular kernels. The coupling of the random batch and the splitting approaches in the RBMC leads us to a highly efficient sampling method.
The splitting of the interaction kernels allows the use of the random batch to obtain a candidate sample with $O(1)$ operations,
and the acceptance-rejection strategy which avoids the long-term relaxation in the case of a particle moving to
repulsive region of a neighboring particle and ensures the detailed balance condition. Since the calculation of the acceptance ratio is only related to the interaction
between a particle and its neighboring particles, the cost remains to be $O(1)$. Theoretical error estimate is present for
the accuracy of the RBMC method.

Numerical examples of the Dyson Brownian motion and the Lennard-Jones fluids demonstrate that the RBMC is
promising in simulating particle systems.
It is remarked that the Lennard-Jones potential itself is of short range in the sense that the $1/r^6$ kernel decays rapidly and
the cutoff can be used. Nevertheless, our numerical results have already shown that the RBMC is faster than the traditional
Metropolis-Hastings sampling using the cutoff scheme. The RBMC will be certainly more advantageous for particle systems with
longer-range interaction such as the Coulomb potential which is essential in many areas of applications. This extension will be studied
in the future.

\section*{Acknowledgement}
The work of L. Li was partially sponsored by NSFC 11901389, Shanghai Sailing Program 19YF1421300 and NSFC 11971314.
The work of Z. Xu and Y. Zhao was partially supported by  NSFC (grant Nos. 11571236 and 21773165) and the HPC center of Shanghai Jiao Tong University.
The authors acknowledge the useful discussions with Professors Shi Jin and Jian-Guo Liu.

\appendix

\section{Proof of Lemma \ref{lmm:detailedbalanceSDE}}\label{app:proof1}

The first claim regarding invariant measure is trivial. Let us only verify the detailed balance formally.
Consider the general SDE:
\[
dX=b(X)\,dt+\sigma(X)dW,
\]
where $\sigma$ is a $d\times m$ matrix while $W$ is an $m$-dimensional standard Wiener process.

The law of $X$ (or the density of the law of $X$), $p(x, t)$, satisfies the Fokker-Planck equation \cite{risken1996fokker}:
\begin{gather}\label{eq:fp}
\partial_t p=-\nabla\cdot(b p)+\frac{1}{2}\sum_{i,j=1}^d\partial_{ij}(\Lambda_{ij}p)=:\mathcal{L}^*p,
\end{gather}
where $\Lambda=\sigma\sigma^T$ and the operator
\[
\mathcal{L}^*=-\nabla\cdot(b \cdot)+\frac{1}{2}\sum_{i,j=1}^d\partial_{ij}(\Lambda_{ij} \cdot)
\]
is the dual operator of the generator of the SDE \cite[Theorem 7.3.3]{oksendal03} given by
\[
\mathcal{L}:=-b\cdot\nabla+\frac{1}{2}\Lambda:\nabla^2.
\]

Let $p(\cdot, t; x_i)$ be the Green's function which is the solution of Eq. \eqref{eq:fp} with
initial condition $p(x, 0; x_i)=\delta(x-x_i)$ and gives the transition density starting from $x_i$:
\[
p(\cdot, t; x_i)= e^{t\mathcal{L}^*}\delta(\cdot-x_i).
\]

The detailed balance condition is therefore the distributional identity
\[
\pi_1(x_i)\mathcal{L}_{z}^*\delta(z-x_i)=\pi_1(z)\mathcal{L}_{x_i}^*\delta(x_i-z)
\]
We pick a test function $\varphi$. It can be shown easily that
\[
\langle \pi_1(x_i)\mathcal{L}^*\delta(\cdot-x_i), \varphi(\cdot)\rangle
=\pi_1(x_i)\mathcal{L}\varphi(x_i),
\]
and that
\[
\langle \pi_1(z)\mathcal{L}_{x_i}^*\delta(x_i-\cdot),\varphi(\cdot)\rangle
=\mathcal{L}_{x_i}^*(\varphi(x_i) \pi_1(x_i)).
\]

Hence, to verify the detailed balance condition, one needs the following
\begin{gather}
\pi_1(x_i)\mathcal{L}\varphi(x_i)=\mathcal{L}_{x_i}^*(\varphi(x_i) \pi_1(x_i)),
\end{gather}
which is reduced to
\[
-\pi_1 b+\frac{1}{2}\nabla\cdot\left(\sigma\sigma^T \pi_1\right)=0.
\]
This holds for $b$ being a gradient field and $\sigma$ being square with $\sigma\propto I$.

\section{Pressure and energy formulas for the Lennard-Jones fluids}\label{app:pressureandenergyLJ}

In computing the total energy or computing the rejection probability for the MH algorithms, one needs to compute the summation $\sum_{i,j}U(x_i-x_j)$. For the purpose, we may make use of the rapid decay for the interaction kernel so that we can pick a truncation distance $r_c$ and approximate the summation for $r_{ij}\ge r_c$ by continuous integral. More precisely,
\begin{gather}
\frac{1}{2}\sum_{i,j: i\neq j}U(x_i-x_j)
\approx \frac{1}{2}\sum_{i\neq j: r_{ij}< r_c}U(x_i-x_j)
+\frac{N}{2}\int_{r_c}^{\infty} 4\pi r^2 \rho(r)u(r) \mathrm{d}r.
\end{gather}
In fact, this formula is obtained by
\begin{multline*}
\frac{1}{2}\sum_{i,j: i\neq j}U(x_i-x_j)
=\frac{1}{2}\sum_{i\neq j: r_{ij}< r_c}U(x_i-x_j)+\frac{1}{2}\sum_{i\neq j: r_{ij}\ge r_c}U(x_i-x_j)\\
=\frac{1}{2}\sum_{i\neq j: r_{ij}< r_c}U(x_i-x_j)
+\frac{1}{2}\sum_{i=1}^N \sum_{j: j\neq i, r_{ij}\ge r_c} u(r_{ij}).
\end{multline*}
Due to symmetry, when the thermal fluctuation is small, $\sum_{j: j\neq i, r_{ij}\ge r_c} u(r_{ij})$ should be independent of $i$, and can be approximated by the following integral using radial density $\rho(r)$:
\[
 \sum_{j: j\neq i, r_{ij}\ge r_c} u(r_{ij})
 \approx \int_{r_c}^{\infty} u(r) 4\pi r^2 \rho(r)\,dr.
\] 

With the approximation of the radial density
\[
\rho(r)\approx \rho=\frac{N}{V},~~r>r_c,
\]
the average tail contribution (per particle) $u^{tail}$ of the interaction energy is given by
\begin{align}
u^{tail}=\dfrac{1}{2}\int_{r_c}^{\infty} 4\pi r^2 \rho u(r) \mathrm{d}r=\dfrac{8}{3}\pi\rho\left[\dfrac{1}{3}\left(\dfrac{1}{r_c}\right)^9-\left(\dfrac{1}{r_c}\right)^3\right].
\end{align}
This tail contribution is independent of the particles so that the acceptance probability in the MH algorithm (Algorithm \ref{alg:mh}) can be approximated by
\begin{gather}\label{appeq:mhrate}
\alpha=\min\left\{1, \exp\left[-\beta\Big(\sum_{j: j\neq i, r_{ij}^*\le r_c}  U(x_i^*-x_j)-\sum_{j: j\neq i, r_{ij}\le r_c} U(x_i-x_j)\Big)\right]  \right\}.
\end{gather}
Apparently, $r_{ij}^*$ means $|x_i^*-x_j|$.

The pressure is calculated by \cite[sec. 3.4]{frenkel2001understanding}
\begin{align}
P=\dfrac{\rho}{\beta}+\dfrac{vir}{V},
\end{align}
where the virial is defined by
\begin{gather}\label{appeq:virial}
\begin{split}
vir=\dfrac{1}{3}\sum\limits_{i}\sum\limits_{j>i}f(r_{ij})\cdot \vec{r}_{ij}.
\end{split}
\end{gather}
The function $f(r_{ij})$ in (\ref{appeq:virial}) is the intermolecular force, i.e., the negative gradient
of the Lennard-Jones potential, so that
\[
vir= 8\sum\limits_{i< j}(2r_{ij}^{-12}-r_{ij}^{-6})
\approx 8\sum\limits_{i< j: r_{ij}\le r_c}(2r_{ij}^{-12}-r_{ij}^{-6})
+4 N\int_{r_c}^{\infty}(2r^{-12}-r^{-6})4\pi r^2\rho\,dr.
\]
Note that we split all summation items into $r_{ij}\geq r_c$ and $r_{ij}<r_c$. Because $f(r_{ij})\cdot \vec{r}_{ij}$
decreases rapidly when $r_{ij}\geq r_c$ and the particles far away can be approximated by an average radial density $\rho=N/V$,
the integral approximation for the summation makes sense.
Consequently, the pressure can be approximated by:
\begin{align}
P=\dfrac{\rho}{\beta}+\frac{8}{V}\sum_i \sum\limits_{j: j>i,|r_{ij}|<r_c}(2r_{ij}^{-12}-r_{ij}^{-6})+\dfrac{16}{3}\pi\rho^2 \left[\dfrac{2}{3}\left(\dfrac{1}{r_c}\right)^9-\left(\dfrac{1}{r_c}\right)^3\right].
\end{align}

With the cutoff, for our proposed algorithm (Alogirthm \ref{fastmcmc}), the cutoff affects the SDE step in the sense that the particles away beyond $r_c$ contribute net zero force to the current particle.  As mentioned in the main text, only when the randomly selected particle is within $B(X^i, r_c)$, we will compute its force and update $X^i$; otherwise, we do not update $X^i$ but still increase $k$ until it reaches $m$.

Consider now the periodic boxes that approximates the fluid, with the length of the box being $L$.  A given particle $x_i$ interacts with all other particles in
the system and their periodic images. In this case, the total potential energy is
\begin{align}
U_{tot}=\dfrac{1}{2}\sum\limits_{i,j,\vec{n}}{'}u(|\vec{r}_{ij}+\vec{n}L|),
\end{align}
where $\vec{n}$ is three-dimensional integer vector and the notation $\sum\limits_{i,j,\vec{n}}{'}$ means that the case $\vec{n}=0$ when $i=j$ are excluded.
We pick the truncation length $r_c=L/2$
as commonly used in literature \cite[Chap. 3]{frenkel2001understanding}. The total energy is then approximated as
\begin{gather}
U_{tot}\approx \dfrac{1}{2}\sum_{|\vec{r}_{ij}+\vec{n}L|< r_c}{'} u(|\vec{r}_{ij}+\vec{n}L|)+\dfrac{8N}{3}\pi\rho\left[\dfrac{1}{3}\left(\dfrac{1}{r_c}\right)^9-\left(\dfrac{1}{r_c}\right)^3\right].
\end{gather}
The pressure and the acceptance rate \eqref{appeq:mhrate} in the MH algorithm are similarly modified for periodic boxes.
The modification for our proposed RBMC algorithm (Algorithm \ref{fastmcmc}) is also similar. We do not repeat here.

\bibliographystyle{SIAM}  
\bibliography{mcsplit}

\begin{thebibliography}{10}

\bibitem{allen1987}
{\sc M.~P. Allen and D.~J. Tildesley}, {\em Computer simulation of liquids},
  Oxford university press, 1987.

\bibitem{besag1994comments}
{\sc J.~E. Besag}, {\em Comments on "{R}epresentations of knowledge in complex
  systems" by {U.} {G}renander and {MI} {M}iller}, J. Roy. Statist. Soc. Ser.
  B, 56 (1994), pp.~591--592.

\bibitem{bottou1998online}
{\sc L.~Bottou}, {\em Online learning and stochastic approximations}, On-line
  learning in neural networks, 17 (1998), p.~142.

\bibitem{bubeck2015convex}
{\sc S.~Bubeck}, {\em Convex optimization: Algorithms and complexity},
  Foundations and Trends{\textregistered} in Machine Learning, 8 (2015),
  pp.~231--357.

\bibitem{dai2016provable}
{\sc B.~Dai, N.~He, H.~Dai, and L.~Song}, {\em Provable {B}ayesian inference
  via particle mirror descent}, in Artificial Intelligence and Statistics,
  2016, pp.~985--994.

\bibitem{dalalyan2017}
{\sc A.~S. Dalalyan}, {\em Theoretical guarantees for approximate sampling from
  smooth and log-concave densities}, Journal of the Royal Statistical Society:
  Series B (Statistical Methodology), 79 (2017), pp.~651--676.

\bibitem{erdos2017}
{\sc L.~Erdos and H.-T. Yau}, {\em Dynamical approach to random matrix theory},
  Courant Lecture Notes in Mathematics, 28 (2017).

\bibitem{frenkel2001understanding}
{\sc D.~Frenkel and B.~Smit}, {\em Understanding molecular simulation: {F}rom
  algorithms to applications}, vol.~1, Elsevier, 2001.

\bibitem{gamerman2006markov}
{\sc D.~Gamerman and H.~F. Lopes}, {\em Markov chain Monte Carlo: stochastic
  simulation for Bayesian inference}, Chapman and Hall/CRC, 2006.

\bibitem{georges1996}
{\sc A.~Georges, G.~Kotliar, W.~Krauth, and M.~J. Rozenberg}, {\em Dynamical
  mean-field theory of strongly correlated fermion systems and the limit of
  infinite dimensions}, Reviews of Modern Physics, 68 (1996), p.~13.

\bibitem{gilks1995markov}
{\sc W.~R. Gilks, S.~Richardson, and D.~Spiegelhalter}, {\em Markov chain
  {Monte Carlo} in practice}, Chapman and Hall/CRC, 1995.

\bibitem{hastings1970monte}
{\sc W.~K. Hastings}, {\em {Monte Carlo} Sampling Methods Using {Markov} Chains
  and Their Applications}, Oxford University Press, 1970.

\bibitem{hetenyi2002multiple}
{\sc B.~Hetenyi, K.~Bernacki, and B.~J. Berne}, {\em Multiple ``time step"
  {M}onte {C}arlo}, J. Chem. Phys., 117 (2002), pp.~8203--8207.

\bibitem{hu2019diffusion}
{\sc W.~Hu, C.~J. Li, L.~Li, and J.-G. Liu}, {\em On the diffusion
  approximation of nonconvex stochastic gradient descent}, Annals of
  Mathematical Sciences and Applications, 4 (2019), pp.~3--32.

\bibitem{Israelachvili::2010}
{\sc J.~N. Israelachvili}, {\em Intermolecular and Surface Forces}, Academic
  Press, 3rd~ed., 2010.

\bibitem{jabin2017}
{\sc P.-E. Jabin and Z.~Wang}, {\em Mean field limit for stochastic particle
  systems}, in Active Particles, Volume 1, Springer, 2017, pp.~379--402.

\bibitem{jin2018random}
{\sc S.~Jin, L.~Li, and J.-G. Liu}, {\em Random batch methods ({RBM}) for
  interacting particle systems}, J. Comput. Phys., 400 (2020), p.~108877.

\bibitem{johnson1993lennard}
{\sc J.~K. Johnson, J.~A. Zollweg, and K.~E. Gubbins}, {\em The
  {L}ennard-{J}ones equation of state revisited}, Molecular Physics, 78 (1993),
  pp.~591--618.

\bibitem{kloeden2013numerical}
{\sc P.~E. Kloeden and E.~Platen}, {\em Numerical solution of stochastic
  differential equations}, vol.~23, Springer Science \& Business Media, 2013.

\bibitem{lasry2007}
{\sc J.-M. Lasry and P.-L. Lions}, {\em Mean field games}, Japanese Journal of
  Mathematics, 2 (2007), pp.~229--260.

\bibitem{li2019stochastic}
{\sc L.~Li, Y.~Li, J.-G. Liu, Z.~Liu, and J.~Lu}, {\em A stochastic version of
  {S}tein {V}ariational {G}radient {D}escent for efficient sampling}, Commun.
  Appl. Math. Comput. Sci., 15 (2020).

\bibitem{li2019mean}
{\sc L.~Li, J.-G. Liu, and P.~Yu}, {\em On the mean field limit for {B}rownian
  particles with {C}oulomb interaction in 3{D}}, J. Math. Phys., 60 (2019),
  p.~111501.

\bibitem{lifshitz2013statistical}
{\sc E.~M. Lifshitz and L.~P. Pitaevskii}, {\em Statistical physics: theory of
  the condensed state}, vol.~9, Elsevier, 2013.

\bibitem{liu2016stein}
{\sc Q.~Liu and D.~Wang}, {\em Stein variational gradient descent: A general
  purpose bayesian inference algorithm}, in Advances In Neural Information
  Processing Systems, 2016, pp.~2378--2386.

\bibitem{ma2015complete}
{\sc Y.-A. Ma, T.~Chen, and E.~Fox}, {\em A complete recipe for stochastic
  gradient {MCMC}}, in Advances in Neural Information Processing Systems, 2015,
  pp.~2917--2925.

\bibitem{Martin:SISC:12}
{\sc J.~Martin, L.~C. Wilcox, C.~Burstedde, and O.~Ghattas}, {\em A stochastic
  newton {MCMC} method for large-scale statistical inverse problems with
  application to seismic inversion}, {SIAM J. Sci. Comput.}, {34} ({2012}),
  pp.~{A1460--A1487}.

\bibitem{martin1998novel}
{\sc M.~G. Martin, B.~Chen, and J.~I. Siepmann}, {\em A novel {M}onte {C}arlo
  algorithm for polarizable force fields: application to a fluctuating charge
  model for water}, J. Chem. Phys., 108 (1998), pp.~3383--3385.

\bibitem{metropolis1953}
{\sc N.~Metropolis, A.~W. Rosenbluth, M.~N. Rosenbluth, A.~H. Teller, and
  E.~Teller}, {\em Equation of state calculations by fast computing machines},
  J. Chem. Phys., 21 (1953), pp.~1087--1092.

\bibitem{milstein2013stochastic}
{\sc G.~N. Milstein and M.~V. Tretyakov}, {\em Stochastic numerics for
  mathematical physics}, Springer Science \& Business Media, 2013.

\bibitem{oksendal03}
{\sc B.~{\O}ksendal}, {\em Stochastic differential equations: an introduction
  with applications}, Springer, Berlin, Heidelberg, sixth~ed., 2003.

\bibitem{raginsky2017non}
{\sc M.~Raginsky, A.~Rakhlin, and M.~Telgarsky}, {\em Non-convex learning via
  stochastic gradient {L}angevin dynamics: a nonasymptotic analysis}, arXiv
  preprint arXiv:1702.03849,  (2017).

\bibitem{rezende2015variational}
{\sc D.~J. Rezende and S.~Mohamed}, {\em Variational inference with normalizing
  flows}, arXiv preprint arXiv:1505.05770,  (2015).

\bibitem{risken1996fokker}
{\sc H.~Risken}, {\em The {F}okker-{P}lanck equation}, Springer, 1996.

\bibitem{robbins1951stochastic}
{\sc H.~Robbins and S.~Monro}, {\em A stochastic approximation method}, The
  Annals of Mathematical Statistics,  (1951), pp.~400--407.

\bibitem{roberts1996}
{\sc G.~O. Roberts and R.~L. Tweedie}, {\em Exponential convergence of
  {L}angevin distributions and their discrete approximations}, Bernoulli, 2
  (1996), pp.~341--363.

\bibitem{santambrogio2015}
{\sc F.~Santambrogio}, {\em Optimal transport for applied mathematicians},
  Birk{\"a}user, NY,  (2015), pp.~99--102.

\bibitem{stanley1971}
{\sc H.~E. Stanley}, {\em Phase transitions and critical phenomena}, Clarendon
  Press, Oxford, 1971.

\bibitem{tao2012}
{\sc T.~Tao}, {\em Topics in random matrix theory}, vol.~132, American
  Mathematical Soc., 2012.

\bibitem{welling2011bayesian}
{\sc M.~Welling and Y.~W. Teh}, {\em Bayesian learning via stochastic gradient
  {L}angevin dynamics}, in Proceedings of the 28th International Conference on
  Machine Learning (ICML-11), 2011, pp.~681--688.

\end{thebibliography}

\end{document}